\documentclass[12pt]{amsart}
\usepackage{amsfonts}
\usepackage{amsthm}
\usepackage{amssymb}
\usepackage{amsmath}
\usepackage{amscd}
\usepackage[latin2]{inputenc}
\usepackage{t1enc}
\usepackage[mathscr]{eucal}
\usepackage{indentfirst}
\usepackage{graphicx}
\usepackage{graphics}
\usepackage{xcolor}
\usepackage{ulem}

\numberwithin{equation}{section}

\theoremstyle{plain}
\newtheorem{Th}{Theorem}%[section]

\newtheorem{Cor}[Th]{Corollary}

\newtheorem{Pro}[Th]{Proposition}

 \theoremstyle{definition}

\newtheorem{?}[Th]{Problem}

\usepackage{soul}

\begin{document}

\title[Structure of distinguishability for a set of states in a GPT]{Realization of an arbitrary structure of perfect distinguishability of states in general probability theory}

\author[M.\ Weiner]{Mih\'aly Weiner}
\address{Department of Analysis, Institute of Mathematics, Budapest University of Technology and Economics,
  M\H{u}egyetem rkp. 3--9 H-1111 Budapest Hungary, and
  MTA-BME Lend\"ulet ``Momentum'' Quantum Information Theory Research Group}
\email{mweiner@math.bme.hu}

\thanks{Supported by the Bolyai J\'anos Fellowship of the Hungarian Academy of Sciences, the \'UNKP-22-5 New National Excellence Program of the Ministry for Innovation and Technology, by the NRDI grant K132097 and by the Ministry of Culture and Innovation and the National Research, Development and Innovation Office within the Quantum Information National Laboratory of Hungary (Grant No. 2022-2.1.1-NL-2022-00004).}

 \subjclass[2010]{}

 \keywords{}

\begin{abstract}
Let $s_1,s_2,\ldots s_n$ be states of a general probability theory, and 
$\mathcal A$ be the set of all subsets of indices $H \subset [n]\equiv\{1,2,\ldots n\}$ such that the states $(s_j)_{j\in H}$ are jointly perfectly distinguishable. All subsets with a single element are of course in $\mathcal A$, and since smaller collections are easier to distinguish, if $H\in \mathcal A$ and $L \subset H$ then $L\in \mathcal A$; in other words, $\mathcal A$ is a so-called {\it independence system} on the set of indices $[n]$. In this paper it is shown that every independence 
system on $[n]$ can be realized in the above manner.
\end{abstract}

\maketitle

\section{Introduction}

General probability theory (GPT)  was conceived in part to have a common framework for both classical and quantum probability theory and in part to explore ``all possibilities that nature could have chosen'' instead of quantum 
probability theory. Perhaps the root of this kind of intention was the study of possible correlations that can arise between two parts of a quantum system. Without any formal GPT structure, one may just consider the set of possible correlations subject to the {\it non-signalling} principle to observe
\cite{tsirelson,popescu_rohrlich} that the ones realizable by a quantum bipartite system form a strict subset. Consequently, it is natural to look for some further physical principle that somehow singles out quantum correlations \cite{infcaus} and to investigate how much better certain tasks could be performed using more general correlations (see e.g.\! \cite{zeroerror,sajat}).

Of course, when we consider a system {\it without} a division into parts, the structure becomes less rich. The concepts we can still talk about are the {\it state space} of the system and the set of possible {\it measurements}, which -- for our purposes -- are simply affine functions from the state space to a simplex (i.e.\! to the state space of a classical system). Since the first of these concepts completely determines\footnote{We assume the {\it no restriction principle} \cite{norestr}: all measurements that are theoretically possible (i.e.\! not ruled out by the convex structure of states) are feasible.} the latter one, ultimately every property of the undivided system is determined by the geometrical shape of the state space; e.g.\! its ``signalling dimension'' \cite{sigdim} or its ``information storability'', which in particular is known to depend on the amount of asymmetry \cite{geoinfcap} of the state space.

In both classical and quantum probability theory, if any two states of a certain collection can be perfectly distinguished, then in fact the whole collection is jointly perfectly distinguishable. However, in general -- when the state space $S$ of the system is not assumed to be neither a simplex (the classical case) or
the set of density operators on $\mathbb C^n$ (the quantum case), it can happen that there are three states $s_1,s_2,s_3\in S$ such that any two are perfectly distinguishable by {\it some} binary measurement, yet there is no measurement with three possible outcomes that would distinguish all of them with zero error.
In fact, in order for the states $s_1,s_2, \ldots s_{n+1} \in S$ to be jointly perfectly distinguishable (or henceforth in short: to be a ``j.p.d.\! collection''), at least they need to form a convex independent set, requiring $S$ to be at least $d=n$ dimensional. This is to be confronted with the fact that for every $d=1,2,\ldots$ one can give a $d$ dimensional state space $S$ and $n=2^d$ elements of $S$ such that any {\it two} can be perfectly distinguished; see \cite{rej}. So clearly, in general the condition of pairwise perfect distinguishability is much weaker than joint perfect distinguishability.

At this point one might wonder what can we say about the structure of perfect distinguishability in general. Let $S$ be a state space of a GPT (i.e.\! any convex set) and fix some elements $s_1,s_2,\ldots s_n\in S$ and consider the set of all subsets of indices 
$$
\mathcal A \equiv \{H\subset [n]\, | \, (s_j)_{j\in H} \; \textrm{is a j.p.d. collection} \}
$$
where we have used the notation $[n]\equiv \{1,2,\ldots n\}$. It is clear that all one element subsets are in $\mathcal A$ and as smaller collections are easier to distinguish, if $H\in \mathcal A$ and $L \subset H$ then $L\in \mathcal A$; so $\mathcal A$ is a so-called {\it independence system} on $[n]$. But is there anything else that can be established at this level of generality? The answer is negative, as we shall seen this by the explicit construction of Section
\ref{sec:constr} which, given an independence system $\mathcal A$ on $[n]$ produces a convex set $S$ and elements $s_1,s_2, \ldots s_{n} \in S$ such that for any $H\subset [n]$ we have that $(s_j)_{j\in H}$  is a j.p.d.\! collection if and only if $H \in \mathcal A$.

This is somewhat similar to the situation of joint measurability, except that for an arbitrary structure of joint measurability it suffices to consider the quantum case (so in the case of joint measurability there is no need to invoke a generic GPT model). A collection of measurements is said to be jointly measurable, if they have a common refinement. Given any independence system $\mathcal A$ on $[n]$, one can produce some measurements $M_1,M_2,\ldots M_n$ on some quantum system such that for any $H\subset [n]$ we have that $(M_j)_{j\in H}$  is a jointly measurable collection if and only if $H \in \mathcal A$; see \cite{joint_measurability} for both the precise definition of joint measurability and the actual construction.

we shall conclude this introduction by pointing out that though perfect distinguishability is a nice mathematical concept, physically it is more relevant to discuss the {\it amount} by which some states can be distinguished. Given some states $s_1,s_2, \ldots s_{n} \in S$, regardless whether they can be jointly perfectly distinguished or not, we might try to find a measurement $M$ with possible outcomes numbered from $1$ to $n$ such that for each $k$, in case the system is in state $s_k$, the measurement $M$ ends up with a ``good enough'' probability with its $k^{\rm th}$ outcome, rightly indicating that the system was in state $s_k$. The usually considered ``figure of merit'', showing how well we can distinguish these states by performing measurement $M$ is the so-called {\it symmetric error probability}, which is the sum (or in some other convention: the average) of the error probabilities; that is, the quantity
$$
\sum_{k=1}^n P({\rm outcome}\neq k|M\, {\textrm{is measured when the system is in state}\; s_k} ).
$$
Following the notations of \cite{milanek}, let us denote by $P_e^*(s_1,s_2\ldots s_n)$ the minimum (or infimum, if there is no minimum) of 
the above quantity when all measurements are considered. By its definition, $P_e^*$ is symmetrical in its arguments, takes values in the interval $[0,1]$ and equals to zero if its arguments form a j.p.d.\! collection. So for a fixed collection of states $s_1,s_2,\ldots s_n$
let us consider the function 
$$ 
[n]\supset H\mapsto F(H):=P_e^*((s_j)_{j\in H}).
$$
Since less states are easier to distinguish, if $L\subset H$, then $F(L)\leq F(H)$. Moreover, given a collection and one more state, we can still perform the optimal measurement for the initial collection (although it will result in a sure error, if the state of system is the added new one), showing that $F(H\cup \{j\})\leq F(H)+1$. However, apart from these trivial properties, what else can be established about $F$, in general? To my knowledge, there is very little we know about this, though for the case when the state space is assumed to be that of a quantum system, \cite{milanek} has a bound on $P_e^*(s_1,s_2\ldots s_n)$ in terms of the pairwise error probabilities $P_e^*(s_j,s_k)$ ($j\neq k$). So both the quantum case and the general one should be topic of further researches.

\section{Preliminaries}
\subsection{The simplex and the affine hyperplane embedding it}

Let $\mathbb R^n_1\subset \mathbb R^n$ denote the affine hyperplane
\begin{equation}
\mathbb R^n_1\equiv \{(t_1,\ldots t_n)\in \mathbb R^n \, | \, t_1+\ldots +t_n=1\}.
\end{equation}
Let further $\Delta_n\subset \mathbb R^n_1$ be the $n$-simplex; i.e.\! the convex set of classical probability distributions with $n$ terms:
\begin{equation}
\Delta_n \equiv \{(p_1,\ldots p_n)\in \mathbb R^n_1 \, |\, \forall j: p_j\geq 0 \}.
\end{equation}
For $m<n$, one might view $\Delta_m$ as a subset of $\Delta_n$ (a probability distribution with $m$ terms can be also viewed as a probability distribution with $n$ terms by appending $n-m$ zero terms to it) and similarly, $\mathbb R^m_1$ as a subset of 
$\mathbb R^n_1$. Then it is meaningful to consider the orthogonal projection $\pi_{n,m}:\mathbb R^n_1 \to \mathbb R^m_1$; we have 
that
\begin{equation}
\label{defofpi}
\pi_{n,m}((t_1,t_2,\ldots, t_k,t_{k+1},\ldots t_n)) = (t_1+r,t_2+r,\ldots t_k+r)
\end{equation}
where $r=(t_{m+1}+\ldots +t_{n})/m$. Note that $\pi_{n,m}:\mathbb R^n_1 \to \mathbb R^m_1$ is an {\it affine} (i.e.\! convex combination preserving) map and that $\pi_{n,m}(\Delta_n)=\Delta_m$.

\subsection{GPT in a nutshell}

To describe a physical system, we want to give a mathematical meaning to the physical notions of {\it state} and {\it measurement} (with finite many possible outcomes). Correspondingly, we shall have a set $S$ which we will refer to as the {\it state space} of the system and whose elements will represent  the possible states of the system, and for each $k=1,2,\ldots$, a set $\mathcal M_k$ whose elements will represent the {\it measurements} with $k$ (numbered) outcomes that we can perform. Given a measurement $M\in \mathcal M_k$ and a state $s\in S$, the model should specify a classical probability distribution $p_{M,s}\in \Delta_k$ which is interpreted as the probability distribution of the outcomes for $M$ if the measurement is performed on the system when it is in state $s$.

From a physical point of view, two states should be regarded the same if for all possible measurements they lead to the same outcome probability distribution. Similarly, two measurement should be regarded the same if in every state they always produce the same outcome probability distribution. Thus, through the evaluation $p$, elements of $S$ and $\mathcal M=\cup_{k=1}^\infty \mathcal M_k$ should {\it separate} each other. In particular, we may identify a state $s$ with the function $p_{\cdot,s}:\mathcal M \to \cup_{k=1}^\infty  \Delta_k$ which in fact, for each $k=1,2,\ldots$, maps $\mathcal M_k$ into $\Delta_k$. As such, it is intrinsically meaningful to consider convex combinations of states; after all, convex combinations of the {\it functions} $p_{\cdot,s_1}, p_{\cdot,s_2}$ are well-defined. Importantly, such a combination is again regarded as a state of the system. Note that in this way, the outcome probability distribution for a fixed measurement, as a function of the state, is -- by definition -- affine. Note also that the convex combination of states has a clear operational meaning: to obtain outcome statistics of the measurement $M$ in the ``mixed state'' $\lambda s_1 +(1-\lambda)s_2$ where $s_1,s_2\in S$ and the coefficients $\lambda,1-\lambda \in [0,1]$, one prepares the system in either state $s_1$ or $s_2$ with corresponding probabilities $\lambda$ and $1-\lambda$, then performs $M$.

Thus by what was explained, we assume that the state space is endowed with a structure making it a {\it convex set} and every measurement is given as an affine map from $S$ to a certain simplex. In what follows, we shall assume the {\it no restriction principle} \cite{norestr}: every ``theoretically possible measurement'', that is, every affine map from $S$ to a simplex, is actually a realizable measurement. Hence the whole probability model is completely determined by a single input: by the choice of the convex set $S$ playing the role of the state space of the physical system.

This principle has some physical motivations and it certainly holds in both the classical and quantum case. In the (finite) classical case $S=\Delta_n$ for some $n$, while in the (finite) quantum case $S$ is the set of density operators on $\mathbb C^n$ for some $n$. As is well-known, $\Phi$ is an affine function from the set of density operators on $\mathbb C^n$ to $\Delta_k$ if and only if there exist $k$ positive operators $E_1,E_2,\ldots E_k\geq 0$ on $\mathbb C^n$ whose sum is the identity, and $\Phi$ is the map
$$
\rho \mapsto ({\rm Tr}(\rho E_1), {\rm Tr}(\rho E_2),\ldots {\rm Tr}(\rho E_k)),
$$
giving back the usual picture where measurements on a quantum system are represented as {\it positive operator valued measures}.

\subsection{Joint perfect distinguishability}
Let $S$ be the state space of a GPT model. A collection of states $s_1,\ldots s_n\in S$ is said to be {\it jointly perfectly distinguishable} or in short: j.p.d., if there exists a measurement with outcomes $1,\ldots n$ such that for all $k=1,\ldots n$, in case the system is in state $s_k$, the measurement ends up with its $k^{\rm th}$ outcome with probability one (i.e.\! rightly indicates that out of the listed states, the system must had been in state $s_k$). Thus, by what was explained in words, $s_1,\ldots s_n\in S$ are j.p.d.\! if and only if there exists an affine map $\Phi:S\to \Delta_n$ such that for all $k=1,\ldots n:$ $(\Phi(s_k))_k=1$. In the quantum case, as is well-known, this requirement is equivalent to asking that the $n$ density operators have (pairwise) diagonal images.

Let $k\leq n$ and $s_1,\ldots s_n\in S$ a j.p.d.\! collection. This means that there is an affine map $\Phi:S\to \Delta_n$ such that $\Phi(j)$ is precisely the $j^{\rm th}$ vertex of $\Delta_n$ for each $j=1,\ldots n$. Then $\pi_{n,k}\circ\Phi$ is an affine map from $S$ to $\Delta_k$ still having the property that $\Phi(j)$ is precisely the $j^{\rm th}$ vertex of $\Delta_k$ for each $j=1,\ldots k$. Hence the smaller collection of states $s_1,\ldots s_k$ is j.p.d., as we naturally expected.

Finally note that the j.p.d.\! requirement for two states $s_1,s_2\in S$ is geometrically equivalent to asking $s_1,s_2$ to be a so-called {\it antipodal} pair of points of the convex set $S$, see \cite{rej}. Thus, the j.p.d.\! property may be viewed as a natural generalization of the mentioned concept giving rise to the notion of ``jointly antipodal collection'' of points. 

\section{The construction}
\label{sec:constr}

Let $n>1$ be an integer and $\mathcal A$ be an independence system on $[n]$. We shall say that a subset 
$H\subset [n]$ is {\it minimally dependent} (with respect to $\mathcal A$), if $H\notin \mathcal A$, but dropping any element of $H$ makes it belong to $\mathcal A$. Note that any minimally dependent set must have at least $2$ elements. Let $\mathcal O_\mathcal A$ be the set of all minimally dependent (with respect to $\mathcal A$) subsets of $[n]$. 

Our construction will start with the vertices of the $n$-simplex; so we shall set 
\begin{equation}
s_1=(1,0,0,\ldots), \; s_2=(0,1,0,\ldots), \, \ldots
\end{equation}
and so on. 
The problem is that if the state space $S$ was simply the convex hull of these points (i.e.\! the simplex $\Delta_n$), 
then $s_1,\ldots s_n$ would be a j.p.d.\! collection. In order to take account of the given independence system $\mathcal A$, for each 
$H\in \mathcal O_{\mathcal A}$ we shall add a number of new points to ``ruin'' the j.p.d.\! property for $(s_j)_{j\in H}$ and hence also for all larger collections, while leaving it ``intact'' for any $(s_j)_{j\in K}$ such that $K$ does not contain $H$.

For a subset $H\subset [n]$ we shall denote by $1_H$ the {\it indicator} of $H$; i.e.\! $1_H$ is the point of $\mathbb R^n$ whose $j^{\rm th}$ coordinate is $1$ if $j\in H$ and zero otherwise. Note that in particular $1_{\{j\}}=s_j$.

Now suppose $H\subset [n]$ has at least $m:=|H|\geq 2$ elements. 
If $H$ is a proper subset (i.e.\! $m<n$) then for every $j\in H$ we shall set
\begin{equation}
\label{defofq}
q_H^{(j)}:= -2\epsilon\, 1_{\{j\}} +\left(\frac{1-\epsilon(m-2)}{m-1} \right)1_{H\setminus\{j\}}+
\left(\frac{\epsilon m}{n-m}\right)1_{[n]\setminus H}
\end{equation}
where the constant $\epsilon=1/(3n^2)$. For the full set $[n]$ instead, we shall set 
\begin{equation}
q_{[n]}^{(j)}:= -\epsilon 1_{\{j\}} + \left(\frac{1+\epsilon}{n-1}\right) 1_{[n]\setminus \{j\}}.
\end{equation}
In both cases, one can easily check that in fact, the thus defined point $q_H^{(j)}\in \mathbb R^n$  belongs to $\mathbb R^n_1$; i.e.\! 
the sum of its coordinates is $1$. However, because its $j^{\rm th}$ coordinate is negative, $q_H^{(j)}\notin \Delta_n$ (but note that all of is other coordinates are positive). Note also that in both cases $-2\epsilon$ is a lower bound on all coordinates and as by a straightforward check 
\begin{equation}
\min\{\left(\frac{1-\epsilon(m-2)}{m-1} \right), \left(\frac{1+\epsilon}{n-1}\right)\} \geq \frac{2}{3n}=2n\epsilon,
\end{equation}
that $2n\epsilon$ is a lower bound on the value of the coordinates whose index belongs to $H\setminus j$. Finally, note that for $1<m<n$, we have that $\frac{1-\epsilon(m-2)}{m-1}> \frac{1}{3n}=\epsilon n > \frac{\epsilon m}{n-m}$, so in case $1<|H|<n$, we have that the coordinate value of $q_H^{(j)}$ associated to an index in $H\setminus \{j\}$ is always greater than that one associated to an index in $[n]\setminus H$.

We shall now introduce our state space of choice. Let $S$ be the convex hull of the set points formed by the vertices $s_j$ $(j\in [n])$ of the simplex $\Delta_n$ together with the newly introduced points $q^{(j)}_H$ for all minimally dependent set $H\in \mathcal O_{\mathcal A}$ and $j\in H$:
\begin{equation}
S:= \textrm{Conv}
    \big\{ \,
    \{s_j|j\in [n]\}
    \cup
    \{q_H^{(j)} | \, 
	H\in \mathcal O_{\mathcal A},\, j\in H \}
	\, 
	\big\}.
\end{equation}

\begin{Pro} For any $H\in \mathcal O_{\mathcal A}$, we have that $(s_j)_{j\in H}$ is {\bf not} a j.p.d. collection of the state space $S$.
\end{Pro}
\begin{proof}
Without loss of generality, we may assume that $H=[m]$ is precisely the first $m$ elements of $[n]$, where $m=|H|\geq 2$. We have to show that there exists no affine function $\Phi:S\to \Delta_m$ that would map each of the points $s_j$ $(j=1,\ldots m)$ into a separate vertex of the simplex $\Delta_m$ (in fact, without loss of generality: maps $s_j$ into the $j^{\rm th}$ vertex of $\Delta_m$, which -- with an abuse of notations -- we shall still denote by $s_j$). 

Since $\Delta_m\subset S$, an affine function $\Phi:S\to\Delta_m$ extends in a unique manner to be a $\hat{\Phi}:\mathbb R^n\to\mathbb R^m$ linear map which preserves the sum of the coordinates (as it must map $\mathbb R^n_1$ into $\mathbb R^m_1$).

So we shall prove our statement by showing that if $\hat{\Phi}:\mathbb R^n\to\mathbb R^m$ is a linear map such that it preserves the sum of coordinates and $\hat{\phi}(s_j)=s_j$ for all $j\in [m]$, then there exists a $j\in [m]$ such that $\hat{\Phi}(q_{[m]}^{(j)})\notin \Delta_m$ and hence that $\hat{\Phi(S)}\not\subset \Delta_m$. We shall treat separately the cases when $m=n$ and when $m<n$.

\smallskip

{\it Case 1: $m=n$}.  The points $s_j$ ($j\in [m]=[n]$) form a basis of $\mathbb R^n$, so the map $\hat{\Phi}$ must be the identity and hence the proof is finished since, as was noted, $q_{[n]}^{(j)}\notin \Delta_n$.

\smallskip
 
{\it Case 2: $2\leq m<n$}. The sum of the coordinates of $\xi:= \hat{\Phi}(1_{[n]\setminus H})$ is equal to that of $1_{[n]\setminus H}$: it is $n-m$. Hence at least one coordinate of $\xi\in \mathbb R^m$ -- say the $j^{\rm th}$ one must be smaller than or equal to $\frac{n-m}{m}$:
\begin{equation}
c:=(\xi)_j\leq \frac{n-m}{m}.
\end{equation}
Since $1_{\{j\}}=s_j$ and 
\begin{equation}
1_{H\setminus \{j\}}=1_{[m]\setminus \{j\}}=\sum_{k\leq m, k\neq j }s_k, 
\end{equation}
taking account of the fact that $\hat{\phi}(s_j)=s_j$ for all $j\in [m]$, a straightforward
substitution of the point $q_H^{(j)}$ given by formula (\ref{defofq}) into $\hat\Phi$ gives
\begin{equation}
\hat{\Phi}(q_{[m]}^{(j)})= -2\epsilon\, s_j +\left(\frac{1-\epsilon(m-2)}{m-1} \right)\sum_{k\leq m, k\neq j} s_k +
\left(\frac{\epsilon m}{n-m}\right)\xi,
\end{equation}
whose $j^{\rm th}$ coordinate is
\begin{equation}
(\hat{\Phi}(q_{[m]}^{(j)}))_j=-2\epsilon + \left(\frac{\epsilon m}{n-m}\right) c \leq -2\epsilon + \epsilon < 0,
\end{equation}
showing that $\hat{\Phi}(q_{[m]}^{(j)})\notin \Delta_m$.
\end{proof}

\begin{Pro} For any $H\in \mathcal A$, we have that $(s_j)_{j\in H}$ is a j.p.d.\! collection of the state space $S$.
\end{Pro}

\begin{proof}
As in the previous proof, we may assume that $H=[m]$ is precisely the first $m$ elements of $[n]$
The statement is trivial if $m<2$ or $m=n$ (in this latter case, $\mathcal O_{\mathcal A}$ is empty and $S=\Delta_n$). 
So suppose that $2\leq m<n$. We know that $\pi_{n,m}$ is affine and $\pi_{n,m}(s_j)=s_j$ for all $j\in [m]$; so if
$\pi_{n,m}(S)\subset \Delta_m$ also holds, then this shows that $(s_j)_{j\in [m]}$ is a j.p.d.\! collection of the 
state space $S$. 

In what follows we shall show that this is indeed the case. We will need to check that the coordinates of
$\pi_{n,m}(q_K^{(j)})$ are all non-negative whenever $K\in \mathcal O_{\mathcal A}$ and $j\in K$. Since 
$\pi_{n,m}(q_K^{(j)}) \in \mathbb R^m_1$, this non-negativity will then show that $\pi_{n,m}(q_K^{(j)}) \in \Delta_m$
whenever $K\in \mathcal O_{\mathcal A}$ and $j\in K$, which in turn will imply that $\pi_{n,m}(S)\subset \Delta_m$.

What remains is a computation of coordinate values of $\pi_{n,m}(q_K^{(j)})$. We will need a case-by-case
treatment depending on the possible position of $K$ and $j$ relative to the subset $H=[m]$. In what follows, 
we shall use the notation $\ell:=|K\setminus [m]|$. Note that $\ell$ must be at least one: 
since $[m]=H\in \mathcal A$ and $K$ is a minimal dependent set, $K$ cannot be a subset of $[m]$.

\smallskip

{\it Case 1: $j\in [m]$}. With the exception of its $j^{\rm th}$ coordinate, all coordinates of $q_K^{(j)}$ are positive. Taking account of the formula (\ref{defofpi}), the definition $q_K^{(j)}$ and the estimates we had on its coordinates and the fact that at least one element of $K\setminus\{j\}$ is outside of $[m]$, it follows that the smallest coordinate of $\pi_{n,m}(q_K^{(j)})$ is the $j^{\rm th}$ one, whose actual value is 
\begin{equation}
(q_K^{(j)})_j + \frac{1}{m}\sum_{b=m+1}^n(q_K^{(j)})_b \geq -2\epsilon + \frac{1}{n} 2n\epsilon = 0.
\end{equation}

{\it Case 2.a: $j \notin [m]$ and $\ell>1$}. 
Apart from $j$, there must be at least one more element of $K$ outside of $[m]$. Hence the sum of coordinates of $q_K^{(j)}$ corresponding to indices in $K\setminus [m]$ is at least $-2\epsilon + 2n \epsilon\geq 0$. Taking account of the fact that coordinates of $q_K^{(j)}$ corresponding to indices in $[m]$ are positive, this implies that all coordinates of $\pi_{n,m}(q_K^{(j)})$ are positive, too.
\smallskip

{\it Case 2.b1: $j \notin [m], \ell=1$ and $[m]\subset K$}.
It is straightforward to check that in this case all coordinates of $\pi_{n,m}(q_K^{(j)})$ are equal and hence they are all of magnitude
$\frac{1}{m}$, which is positive.  
\smallskip

{\it Case 2.b2: $j \notin [m], \ell=1$ and $[m]\not\subset K$}. Let $k:=|K|$; we have that $|K\cup [m]|=k-\ell=k-1<m$, 
since $[m]\not\subset K$. The value of every coordinate of $q_K^{(j)}$ associated to an index in $[n]\setminus ([m]\cup \{j\})$
is $\frac{\epsilon k}{n-k}$, so the sum of the coordinates of $q_K^{(j)}$ associated to indices in $[n]\setminus [m]$ is
\begin{equation}
-2 \epsilon + (n-m-1) \frac{\epsilon k}{n-k}.
\end{equation}
On the other hand, the minimal value of the coordinates of $q_K^{(j)}$ associated to indices in $[m]$ is also $\frac{\epsilon k}{n-k}$.
Thus the minimal coordinate value of $\pi_{n,m}(q_K^{(j)})$ is
\begin{equation}
\frac{\epsilon k}{n-k} + \frac{1}{m}\left(-2 \epsilon + (n-m-1) \frac{\epsilon k}{n-k}\right).
\end{equation}
After some rearrangement, multiplying the above expression by $\frac{m(n-k)}{\epsilon}$ makes $\epsilon$ and $m$
disappear from it and yields
\begin{equation}
k(n-1)-2(n-k)
\end{equation}
which is evidently positive, since $n>1$ and $k\geq 2$.
\smallskip

We exhausted all possibilities, so the proof is complete.  
\end{proof}

\begin{Cor}
For any $H\subset [n]$, we have that $(s_j)_{j\in H}$ is a j.p.d.\! collection of the state space $S$
if and only if $H\in \mathcal A$.
\end{Cor}

\end{document}